\documentclass[a4paper,UKenglish,cleveref, autoref, thm-restate]{lipics-v2021}
\usepackage{graphicx} 

\title{Matching Statistics speed up BWT construction} 

\author{Francesco Masillo}{Department of Computer Science, University of Verona, Italy}{francesco.masillo@univr.it}{https://orcid.org/0000-0002-2078-6835}{}

\authorrunning{Francesco Masillo} 

\Copyright{Francesco Masillo} 



\category{} 

\relatedversion{} 

\supplement{\url{https://github.com/fmasillo/CMS-BWT}}

\acknowledgements{I want to thank Sara Giuliani for listening and discussing the preliminary ideas contained in this paper. I also want to thank Zsuzsanna Lipt{\'a}k for giving helpful feedback during the writing of this paper.} 

\nolinenumbers 

\hideLIPIcs  

\EventEditors{John Q. Open and Joan R. Access}
\EventNoEds{2}
\EventLongTitle{42nd Conference on Very Important Topics (CVIT 2016)}
\EventShortTitle{CVIT 2016}
\EventAcronym{CVIT}
\EventYear{2016}
\EventDate{December 24--27, 2016}
\EventLocation{Little Whinging, United Kingdom}
\EventLogo{}
\SeriesVolume{42}
\ArticleNo{23}

\usepackage{multicol,multirow}
\usepackage{amsthm}
\usepackage{latexsym,amsmath,amssymb,stmaryrd,mathtools}
\usepackage{graphicx,epsfig,tikz}
\usepackage{float}
\usepackage[titlenumbered,ruled,linesnumbered]{algorithm2e}
\usetikzlibrary{shapes.geometric,positioning,chains,fit,calc}
\tikzset{near start abs/.style={xshift=1cm},
	node/.style={circle,draw},
	nodeone/.style={circle,draw,gray, line width=0.7mm},
	nodered/.style={circle,draw,red, line width=0.7mm}}

\usepackage[useregional]{datetime2}

\newcommand{\BWT}{\text{BWT}}
\newcommand{\Oh}{{\cal O}}
\newcommand{\Ch}{{\cal C}}
\newcommand{\SA}{\textit{SA}}

\newcommand{\ISA}{\textit{ISA}}
\newcommand{\LCP}{\textit{LCP}}

\newcommand{\RMQ}{\textit{RMQ}}
\newcommand{\MS}{\textit{MS}}

\newcommand{\Suf}{\textit{suf}}
\newcommand{\Pref}{\textit{pref}}
\newcommand{\PSV}{\textit{PSV}}
\newcommand{\NSV}{\textit{NSV}}

\newcommand{\pred}{\textit{pred}}

\newcommand{\CMS}{\textit{CMS}}
\newcommand{\eCMS}{\textit{eCMS}}

\newcommand{\lcp}{\textit{lcp}}
\newcommand{\ems}{\textit{ems}}
\newcommand{\ip}{\textit{ip}}

\newcommand{\head}{\text{head}}
\newcommand{\ihead}{\text{insert-head}}
\newcommand{\bigBWT}{{\tt big}-{\tt BWT}}
\newcommand{\CMSBWT}{{\tt CMS}-{\tt BWT}}
\newcommand{\grlBWT}{{\tt grlBWT}}
\newcommand{\rPFP}{{\tt r}-{\tt pfBWT}}
\newcommand{\ropeBWT}{{\tt ropeBWT2}}
\newcommand{\libsais}{{\tt libsais}}

\definecolor{darkred}{RGB}{128,0,0}
\definecolor{darkgreen}{RGB}{0,128,0}
\definecolor{lightgreen}{RGB}{224,255,224}
\definecolor{darkblue}{RGB}{0,0,128}
\definecolor{lightblue}{RGB}{224,244,255}
\newcommand{\colorpar}[3]{\colorbox{#1}{\parbox{#2}{#3}}}
\newcommand{\marginremark}[3]{\marginpar{\colorpar{#2}{\linewidth}{\color{#1}#3}}}
\newcommand{\remarkF}[1]{\marginremark{blue}{lightblue}{\tiny{[F]~ #1}}}

\renewcommand{\epsilon}{\varepsilon}

\begin{document}

\maketitle

\begin{abstract}
Due to the exponential growth of genomic data, constructing dedicated data structures has become the principal bottleneck in common bioinformatics applications. In particular, the Burrows-Wheeler Transform (\BWT) is the basis of some of the most popular self-indexes for genomic data, due to its known favourable behaviour on repetitive data.

Some tools that exploit the intrinsic repetitiveness of biological data have risen in popularity, due to their speed and low space consumption. 
We introduce a new algorithm for computing the $\BWT$, which takes advantage of the redundancy of the data through a compressed version of matching statistics, the $\CMS$ of [Lipt{\'a}k et al., WABI 2022]. We show that it suffices to sort a small subset of suffixes, lowering both computation time and space. Our result is due to a new insight which links the so-called insert-heads of [Lipt{\'a}k et al., WABI 2022] to the well-known run boundaries of the $\BWT$.

We give two implementations of our algorithm, called \CMSBWT, both competitive in our experimental validation on highly repetitive real-life datasets. In most cases, they outperform other tools w.r.t.\ running time, trading off a higher memory footprint, which, however, is still considerably smaller than the total size of the input data.

\end{abstract}

\section{Introduction}\label{sec:introduction}
The Burrows-Wheeler Transform ($\BWT$)~\cite{BurrowsWheeler94} is a reversible permutation of the characters of the input text that can be computed in linear time and space. It is closely related to the suffix array, a permutation of the indices of $T$ which is based on the lexicographic order of the suffixes. It is known that the \BWT, when applied to repetitive texts, results in an easier-to-compress string, especially when using simple run-length encoding. 

Another virtue of the $\BWT$ is that it can be used as a self-index to replace the original text. It can support pattern matching queries, more specifically, it counts how many times a pattern $P$ occurs in $T$ using time proportional to $|P|$. It can be incorporated into more elaborate indexes~\cite{FerraginaM00, GagieNP20, GiulianiR022, RossiOLGB22} to support locating queries (finding the positions in $T$ where $P$ occurs) and more complex queries such as finding Maximal Exact Matches (MEMs) and Maximal Unique Matches (MUMs). MEMs and MUMs are of key importance, especially in the field of bioinformatics, where they are used for read alignment (e.g.\ MUMmer~\cite{MarcaisDPCSZ18}). Widely used tools such as BowTie2~\cite{LangmeadS12} and BWA~\cite{LiD10} are based on the $\BWT$ for aligning short reads to a reference genome.

Nowadays, the amount of biological data that is publicly available is too big for most $\BWT$ construction algorithms. A key observation is that, even though the size of these datasets is massive, the information contained within them is highly redundant. This is why tools such as \bigBWT~\cite{BoucherGKLMM19}, \rPFP~\cite{Oliva2023} and \grlBWT~\cite{Diaz-DominguezN22} have emerged. These tools exploit the intrinsic repetitiveness of the input data to build large $\BWT$s fast and in compressed space.

Recently, the authors of~\cite{LiptakMP22} 
devised a variant of matching statistics~\cite{ChangL94} called {\em compressed matching statistics}. This data structure has been proven to be effective in expressing the redundancy of sets of highly similar strings and being used in the process of suffix sorting. In this paper, we are going to show that this compressed representation of matching statistics can also be highly useful in building large $\BWT$s, due to the fact that it uses considerably less space than the input data. 

Experimental results show that our implementation, \CMSBWT, is competitive, if not better, than the state-of-the-art tools for constructing the $\BWT$, although heavier on the space consumption side.

Our algorithm is based on a new insight that allows us to find the run boundaries of the $\BWT$ within special buckets of suffixes, which are closely connected to the fundamental element of the compressed matching statistics, the so-called insert-heads.

The paper is organized as follows. In Section~\ref{sec:basics}, we give definitions and notations used in the remainder of the paper. Section~\ref{sec:cms} contains an overview of the compressed matching statistics. In Section~\ref{sec:bwt-construction}, we present our contribution, describing the algorithm used for constructing the $\BWT$. In Section~\ref{sec:implementation}, we describe details of our implementation and then report experimental results in Section~\ref{sec:experiments}. Finally, in Section~\ref{sec:conclusions}, conclusions and future work are discussed.

\section{Basics}\label{sec:basics}
Let $\Sigma$ be an ordered alphabet of size $\sigma$. A string $T$ over $\Sigma$ is a finite sequence of characters from $\Sigma$. The $i$th character of $T$ is denoted $T[i]$, its length is $|T| = n$, and $T[i..j]$ denotes the substring $T[i]\cdots T[j]$. If $i > j$, then $T[i..j]$ is the empty string $\epsilon$. The suffix $T[i..]=T[i..n]$ is referred to as the $i^{th}$ suffix $\Suf_i(T)$, and $T[..i]=T[1..i]$ is the $i^{th}$ prefix $\Pref_i(T)$. When T is clear from the context, we write $\Suf_i$ for $\Suf_i(T)$.  

We assume that the last character of $T$ is the sentinel character \$. It is set to be smaller than any other character in $\Sigma$ and appears only once as the end-of-string character.

The {\em suffix array} $\SA$ of a string $T$ is a permutation of the set $\{1,\ldots,n\}$ such that $\SA[i]=j$ if $\Suf_j(T)$ is the $i$th in lexicographic order among all suffixes. Numerous suffix array construction algorithms (SACAs) exist in the literature~\cite{NongZC11,Nong13,Baier16,LiLH22,Goto19}. SA-IS~\cite{NongZC11} is by far the most popular linear time SACA, being both simple and fast in practice. 

The {\em inverse suffix array} $\ISA$ is the inverse permutation of $\SA$. 

The {\em longest common prefix} ($\lcp$) of a pair of strings $T$ and $S$ is the longest string $U$ which is the prefix of both $T$ and $S$. The {\em longest-common-prefix array} $\LCP$ is another array closely related to the $\SA$. It is given by: $\LCP[1] = 0$, and for $i > 1$, $\LCP[i]$ is the length of the longest common prefix of the two suffixes $\Suf_{\SA[i-1]}$ and $\Suf_{\SA[i]}$. This array can be computed in linear time, too~\cite{KMP09}. 


The Burrows-Wheeler Transform $\BWT$~\cite{BurrowsWheeler94} is a reversible permutation of the input text $T$. It is defined as $\BWT[i] = \$$ if $\SA[i] = 1$, and $\BWT[i] = T[\SA[i] - 1]$ otherwise.

Let $R$ and $S$ be two strings. The {\em matching statistics} $\MS$ of $S$ w.r.t.\ $R$ is an array of length $|S|$ in which every entry is a pair of integers defined as follows. Fix $i$, let $U_i$ be the longest prefix of suffix $\Suf_i(S)$ which occurs as a substring in $R$. Then, entry $\MS[i] = (p_i, \ell_i)$, where $p_i$ is an occurrence of $U_i$ in $R$, or $-1$ if $U_i = \epsilon$, and $\ell_i = |U_i|$. We will call $U_i$ {\em matching factor} and the character $c_i = S[i+\ell_i]$ will be referred to as {\em mismatch character} of position $i$. We set the end-of-string character of $R$ to be smaller than the one of $S$ ($\# < \$$).

For an integer array $A$ of length $n$ and an index $i$, the previous and next smaller values are defined as follows: $\PSV(A,i) = \max\{i' < i\ :\ A[i'] < A[i] \}$, $\NSV(A,i) = \min\{i' > i\ :\ A[i'] < A[i] \}$. The minimum of the empty set is $-\infty$ and the maximum is $+\infty$. There exists a data structure of size $n\log(3+2\sqrt{2}) + o(n)$ bits that can be built in $\Oh(n)$ time and answers both $\PSV$ and $\NSV$ queries in constant time~\cite{Fischer11}.

Given a set of integers $X$ and an integer $x$, the predecessor of $x$ is the largest element in $X$ less than or equal to $x$. In other words, $\pred_X(x) = \max\{y \in X\ :\ y \leq x\}$. Predecessor queries can be answered in $\Oh(\log\log|X|)$ time using the y-fast trie data structure of Willard~\cite{Willard83} which uses $\Oh(|X|)$ space. 

Let $\Ch = \{S_1, S_2, \ldots, S_m\}$ be a collection of strings  not necessarily distinct, i.e.\ $\Ch$ is a multiset. The total length of $\Ch$ will be denoted by $N$, where we use end-of-string characters to delimit the strings, i.e.\ $N = \sum_{d=1}^m |S_d| + m$. From now on, we will treat $\Ch$ as this concatenated string, slightly abusing notation. 

\medskip
Our problem is defined as follows:
\begin{quote}
    {\bf Problem Statement:} Given a string collection $\Ch = \{ S_1,\ldots, S_m\}$ and a reference string $R$, compute the Burrows-Wheeler Transform $\BWT$ of $\Ch$. 
\end{quote}

The end-of-string character \# of $R$ is assumed to be smaller than any of $\Ch$. Moreover, in our setting, we assume that each $S_i \in \Ch$ is highly similar to $R$.

\section{Compressed Matching Statistics}\label{sec:cms}
Recently, the authors of~\cite{LiptakMP22} introduced a new data structure called Compressed Matching Statistics ($\CMS$). This data structure exploits the redundancy of plain $\MS$, where we have the following property: if $\ell_{i} > 0$, then $\ell_{i+1} \geq \ell_i - 1$. We can identify sequences of the form $(x, x-1, x-2, \ldots)$ where $x = \ell_i$, called {\em decrement runs}. A decrement run ends when $\ell_{j} > \ell_{j-1} - 1$, and $j$ is the starting position of a $\head$. For an example see Figure~\ref{fig:example-cms-ecms}.

\begin{definition}[Compressed matching statistics]~\cite{LiptakMP22}
Let $R,S$ be two strings over $\Sigma$, and $\MS$ be the matching statistics of $S$ w.r.t.\ $R$. The {\em compressed matching statistics (\CMS) of $S$ w.r.t.\ $R$} is a data structure storing $(j,\MS[j])$ for each head $j$, and 
a predecessor data structure on the set of heads $H$.  
\end{definition}

It was shown in~\cite{LiptakMP22} that it is possible to recover each individual value $\MS[i]$ for any $i$ using the following formula: $\MS[i] = (p_i + k, \ell_i - k)$, where $j = \pred_H(i)$ and $k = j - i$. This can be done in $\Oh(\log\log\chi)$ time and $\Oh(\chi)$ space, where $\chi = |H|$.

It was shown in~\cite{LiptakMP22} that storing the matching statistics information only for heads leads to a compression ratio of up to 100 times on real-life data.

\subsection{Enhanced Compressed Matching Statistics}
In~\cite{LiptakMP22}, the $\CMS$ was refined with additional information to get the {\em enhanced compressed matching statistics} (\eCMS). Assuming that all characters occurring in $S$ also occur in $R$ at least once, the information of $p_i$ can be made more specific, namely, one can compute the position that a suffix from $S$ would have if it was present in $\SA_R$ the $\SA$ of $R$. This position is called {\em insert point} of $i$:

\begin{align*} 
ip(i) = 
\begin{cases} 1 & \text{ if $U_i=\epsilon$},\\  
\max\{ j \mid U_i \text{ is a prefix of } R[\SA_R[j]..] \text{ and } R[\SA_R[j]..] < U_ic\} & \text{ if this set is non-empty,}\\
 \min\{j \mid U_i \text{ is a prefix of } R[\SA_R[j]..]\} & \text{ otherwise.} 
 \end{cases} 
\end{align*}

The first case is satisfied only for the end-of-string characters of the collection $\Ch$, because the sentinel character of $R$ is smaller than any other character ($\# < \$$). In the other two cases, the insert point is the lexicographic rank of suffix $i$ among all suffixes of $R$. Suffix $i$ ideally points to the next smaller occurrence of $U_i$ in $R$, if it exists (case 2). Otherwise, it coincides with the smallest occurrence of $U_i$ in $R$ (case 3).

For the $\eCMS$, the positions for which the $\MS$ information is saved are called {\em insert-heads} and are defined as follows: $j$ is an insert-head if $\SA_R[\ip(j)] \neq \SA_R[\ip(j-1)] + 1$.
Some additional information is also stored in each $\ihead$: $c_i$, the mismatching character, and $x_i$, a boolean value associated with $c_i$. This value is set to be smaller (S = 0) if $c_i < R[\SA[\ip(i)] + \ell_i]$ or larger (L = 1) otherwise. Referring to the definition of $\ip$, $x_i=1$ whenever we are in case 2, $x_i=0$ when we are in case 3.

\begin{definition}[Enhanced compressed matching statistics]~\cite{LiptakMP22}
Let $R,S$ be two strings over $\Sigma$. Define the {\em enhanced matching statistics} of $S$ w.r.t.\ $R$ as follows: for $1\leq i \leq |S|$, let $\ems(i) = (q_i,\ell_i,x_i,c_i)$, where $q_i = \SA_R[ip(i)]$, $\ell_i$ is the length of the matching factor $U$ of $i$, $c_i$ is the mismatch character, and $x_i\in \{S,L\}$ indicates whether $U_ic_i$ is smaller (S) or greater (L) than $R[q_i..]$. 
The {\em enhanced compressed matching statistics (\eCMS) of $S$ w.r.t.\ $R$} is a data structure storing $(j,\ems(j))$  for each insert-head $j$, and a predecessor data structure on the set of insert-heads $K$. 
\end{definition}

The size of $K$ is denoted by $|K|=\kappa$. The time for recovering $\MS[i]$ becomes $\Oh(\log\log\kappa)$, while the space becomes $\Oh(\kappa)$~\cite{LiptakMP22}.

By definition, the number of $\ihead$s is larger than the number of $\head$s. Although in~\cite{LiptakMP22} the difference in numbers is noticeable, the compression effect is still very strong. For actual numbers see Section~\ref{subsec:datasets}, more specifically Table~\ref{tab:datasets}.

For an example of $\eCMS$ refer to Figure~\ref{fig:example-cms-ecms}.

\begin{figure}
    \centering
    \begin{tabular}{|r|c|c|c|c|c|c|c|c|c|c|c|c|}
         \hline
         $i$ & 1 & 2 & 3 & 4 & 5 & 6 & 7 & 8 & 9 & 10 & 11 & 12 \\
         $R$ & {\tt C} & {\tt A} & {\tt T} & {\tt T} & {\tt A} & {\tt G} & {\tt A} & {\tt T} & {\tt T} & {\tt A} & {\tt G} & {\tt \#}  \\
         \hline 
         \hline
         $S$ & {\tt T} & {\tt A} & {\tt G} & {\tt A} & {\tt G} & {\tt A} & {\tt T} & {\tt T} & {\tt A} & {\tt T} & {\tt T} & {\tt \$}  \\
         $p_i$ & 4 & 5 & 6 & 5 & 6 & 7 & 8 & 9 & 2 & 3 & 4 & -1\\
         $\ell_i$ & 4 & 3 & 2 & 6 & 5 & 4 & 3 & 2 & 3 & 2 & 1 & 0 \\
         head & \checkmark & & & \checkmark & & & & & \checkmark & & & \\
         \hline
         $q_i$ & 4 & 5 & 6 & 5 & 6 & 2 & 3 & 4 & 7 & 8 & 9 & 12 \\
         \begin{tabular}{@{}r@{}}insert- \\ head\end{tabular} & \checkmark & & & \checkmark & & \checkmark & & & \checkmark & & & \checkmark \\
         $c_i$ & {\tt G} & & & {\tt T} & & {\tt T} & & & {\tt \$} & & & {\tt \$} \\ 
         $x_i$ & S & & & L & & L & & & S & & & L \\
         \hline
    \end{tabular}
    \caption{An example for matching statistics and corresponding $\CMS$ and $\eCMS$. In rows 1 and 2, we report $\MS[i] = (p_i, \ell_i)$ of $S$ w.r.t.\ $R$. In row 3, we mark the starting positions of the heads ($\CMS$). In row 4, for each index, we give the special position $q_i = \SA_R[\ip(i)]$, where $\ip(i)$ is the insert point of $\Suf_i$ in $\SA_R$. In row 5, we mark insert-heads ($\eCMS$). In the last two rows, we complete the information stored in $\ihead$s, namely the mismatch character $c_i$ and $x_i$, the associated boolean value (S = smaller, L = larger).}
    \label{fig:example-cms-ecms}
\end{figure}

\subsection{Comparing two suffixes using eCMS}\label{subsec:compare-suf}
The additional information of $\ihead$s helps bucketing suffixes with respect to the insert point. We will call these buckets {\em insert-buckets}. Assessing the order of any two suffixes having different insert point has been proven in the following lemma:

\begin{lemma}~\cite{LiptakMP22}\label{lemma:sufcomp_1}
Let $1\leq i, j\leq N$. If $ip(i) < ip(j)$, then $\Suf_i < \Suf_{j}$.
\end{lemma}

On the other hand, when two suffixes belonging to the same insert-bucket are compared the following lemma refines the order:

\begin{lemma}~\cite{LiptakMP22}\label{lemma:sufcomp_2}
Let $1\leq i, j\leq N$, and $ip(i) = ip(j)$. 
\begin{enumerate}
    \item If $\ell_{i} < \ell_{j}$ and $x_{i}=S$, then $\Suf_i < \Suf_j$. 
    \item If $\ell_{i} < \ell_{j}$ and $x_{i}=L$, then $\Suf_j < \Suf_i$. 
    \item If $\ell_{i} = \ell_{j}$ and $x_{i}=S$ and $x_{j}=L$, then $\Suf_i < \Suf_j$. 
    \item If $\ell_{i} = \ell_{j}$ and $x_{i}=x_{j}$ and $c_{i}<c_{j}$, then $\Suf_i < \Suf_j$. 
\end{enumerate}
\end{lemma}

To achieve the final correct order of two suffixes having the same $\ihead$ information, in~\cite{LiptakMP22} it was suggested sorting only the $\ihead$s. Then, using the new rank for each head, the total order of two suffixes can be established. 

If two arbitrary suffixes from $\Ch$ are being compared, one needs to perform two predecessor queries to get the $\ihead$ of each suffix. This implies that the time spent for a single comparison is $\Oh(\log\log\kappa)$. As we will see in Section~\ref{sec:bwt-construction}, we can avoid the predecessor queries when scanning the collection left to right, resulting in constant time comparisons. This is because we will only perform comparisons of suffixes of one $\ihead$ at a time with other $\ihead$s of the same insert-bucket. 

\subsection{Computing the eCMS}\label{subsec:compute-eCMS}
We will use the procedure outlined in~\cite{LiptakMP22}.

The data structures needed to compute the $\eCMS$ of $\Ch$ w.r.t.\ $R$ are the suffix array $\SA_R$, the inverse suffix array $\ISA_R$, the $\LCP$-array $\LCP_R$, and the $\RMQ$ data structure for $\PSV$-$\NSV$ queries on $\LCP_R$. Every data structure can be constructed in $\Oh(|R|)$ time and space.

This procedure takes $\Oh(N\log|R|)$ time and $\Oh(|R|)$ space and outputs the set of insert-heads of size $\Oh(\kappa)$.

To speed up the practical running time, we will use also the following proposed heuristic of~\cite{LiptakMP22}. Because we work with highly similar strings, it is common to have a singleton interval (an interval of size one) after the failure of a sequence of right extensions. A key insight is that also after the subsequent left contraction, the interval remains of size one. This means that the matching factor $U_i$ lies within a leaf branch in a hypothetical suffix tree of $R$. In order to detect these cases, we can compare $\ell_i$ to the maximum value in $\LCP_R$. If $\ell_i - 1 > \max(\LCP_R)$, then it means that there is no other suffix in $R$ with a prefix equal to $U_i$. This means that we are in a leaf branch.
Computing the left contraction is now equal to accessing $\ISA[p_i + 1]$. This bypasses the $\PSV$ and $\NSV$ queries on $\LCP_R$, avoiding the corresponding cache misses. 
Computing a single maximum can be too restricting for some datasets, so a refinement of this strategy is to divide $\LCP_R$ into blocks and compute a maximum value for each of them. 

The practical speedup can be of an order of magnitude when using this last strategy on sets of highly repetitive strings, as it was shown in~\cite{LiptakMP22}.

\section{Computing the BWT with enhanced Compressed Matching Statistics}\label{sec:bwt-construction}

In this section, we are going to outline the procedure used to compute the $\BWT$ of $\Ch$ using only data structures built on $R$ and the $\eCMS$ of $\Ch$ w.r.t.\ $R$.

We will use the following heuristic in order to speed up the computation of $\BWT(\Ch)$: suffixes in text order between two $\ihead$s are preceded by the same character present in the reference. This can be intuitively explained by looking at the way $\eCMS$ is built: any position between two consecutive $\ihead$s is consecutive in text order both in $R$ and $\Ch$. Therefore, by knowing the insert point of each suffix we know what its position is in the previously computed $\SA_R$, and consequently which is the preceding character stored in $\BWT_R$. This insight tells us that we just have to ``expand'' $\BWT_R$ based on the number of suffixes with the same insert-point while taking care of $\ihead$s. The suffixes corresponding to the starting positions of $\ihead$s are the only ones that need to be sorted inside each insert-bucket. See Figure~\ref{fig:example-head-vs-suf} for an example.

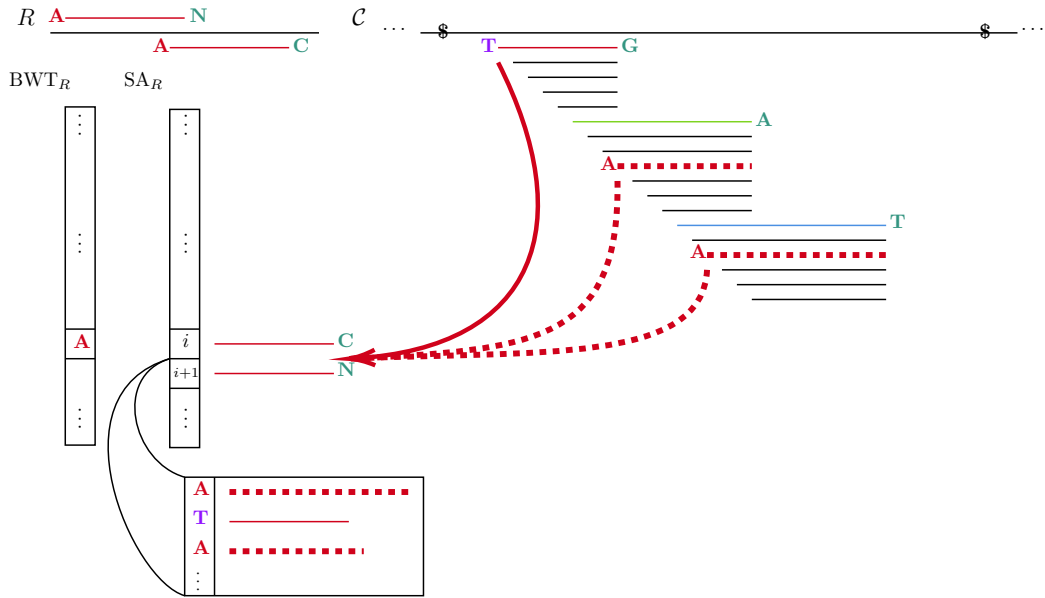
\begin{figure}
    \centering
    \tikzset{every picture/.style={line width=0.75pt}} 

    \scalebox{0.745}{
    \tikzset{every picture/.style={line width=0.75pt}} 
    
    \begin{tikzpicture}[x=0.75pt,y=0.75pt,yscale=-1,xscale=1]
    
    \draw    (30,20) -- (210,20) ;
    \draw    (278,20) -- (678,20) ;
    \draw   (40,70) -- (60,70) -- (60,298.33) -- (40,298.33) -- cycle ;
    \draw   (110,71.67) -- (130,71.67) -- (130,300) -- (110,300) -- cycle ;
    \draw    (110,220) -- (130,220) ;
    \draw    (110,240) -- (130,240) ;
    \draw    (110,260) -- (130,260) ;
    \draw    (40,220) -- (60,220) ;
    \draw    (40,240) -- (60,240) ;
    \draw [color={rgb, 255:red, 208; green, 2; blue, 27 }  ,draw opacity=1 ]   (330,30) -- (410,30) ;
    \draw [color={rgb, 255:red, 208; green, 2; blue, 27 }  ,draw opacity=1 ]   (40,10) -- (120,10) ;
    \draw [color={rgb, 255:red, 208; green, 2; blue, 27 }  ,draw opacity=1 ]   (110,30) -- (190,30) ;
    \draw [color={rgb, 255:red, 208; green, 2; blue, 27 }  ,draw opacity=1 ]   (140,230) -- (220,230) ;
    \draw [color={rgb, 255:red, 208; green, 2; blue, 27 }  ,draw opacity=1 ]   (140,250) -- (220,250) ;
    \draw    (340,40) -- (410,40) ;
    \draw    (350,50) -- (410,50) ;
    \draw    (360,60) -- (410,60) ;
    \draw    (370,70) -- (410,70) ;
    \draw [color={rgb, 255:red, 126; green, 211; blue, 33 }  ,draw opacity=1 ]   (380,80) -- (500,80) ;
    \draw    (390,90) -- (500,90) ;
    \draw    (400,100) -- (500,100) ;
    \draw [color={rgb, 255:red, 208; green, 2; blue, 27 }  ,draw opacity=1 ][line width=3]  [dash pattern={on 3.38pt off 3.27pt}]  (410,110) -- (500,110) ;
    \draw    (420,120) -- (500,120) ;
    \draw    (440,140) -- (500,140) ;
    \draw    (430,130) -- (500,130) ;
    \draw [color={rgb, 255:red, 74; green, 144; blue, 226 }  ,draw opacity=1 ]   (450,150) -- (590,150) ;
    \draw    (460,160) -- (590,160) ;
    \draw    (480,180) -- (590,180) ;
    \draw    (500,200) -- (590,200) ;
    \draw    (490,190) -- (590,190) ;
    \draw [color={rgb, 255:red, 208; green, 2; blue, 27 }  ,draw opacity=1 ][line width=3]  [dash pattern={on 3.38pt off 3.27pt}]  (470,170) -- (590,170) ;
    \draw [color={rgb, 255:red, 208; green, 2; blue, 27 }  ,draw opacity=1 ][line width=3]  [dash pattern={on 3.38pt off 3.27pt}]  (410,120) .. controls (411.67,228.33) and (356.33,234.33) .. (230,240) ;
    \draw [color={rgb, 255:red, 208; green, 2; blue, 27 }  ,draw opacity=1 ][line width=3]  [dash pattern={on 3.38pt off 3.27pt}]  (470,180) .. controls (467,244.33) and (356.33,234.33) .. (230,240) ;
    \draw [color={rgb, 255:red, 208; green, 2; blue, 27 }  ,draw opacity=1 ][line width=2.25]    (330,40) .. controls (384.45,148.57) and (356.9,232.64) .. (233.76,239.81) ;
    \draw [shift={(230,240)}, rotate = 357.43] [color={rgb, 255:red, 208; green, 2; blue, 27 }  ,draw opacity=1 ][line width=2.25]    (17.49,-5.26) .. controls (11.12,-2.23) and (5.29,-0.48) .. (0,0) .. controls (5.29,0.48) and (11.12,2.23) .. (17.49,5.26)   ;
    \draw    (120,400) .. controls (84.33,389.67) and (30.33,263.67) .. (110,240) ;
    \draw    (120,320) .. controls (84.33,309.67) and (71.67,251.67) .. (110,240) ;
    \draw   (120,320) -- (280,320) -- (280,400) -- (120,400) -- cycle ;
    \draw    (140,320) -- (140,400) ;
    \draw [color={rgb, 255:red, 208; green, 2; blue, 27 }  ,draw opacity=1 ][line width=3]  [dash pattern={on 3.38pt off 3.27pt}]  (150,330) -- (270,330) ;
    \draw [color={rgb, 255:red, 208; green, 2; blue, 27 }  ,draw opacity=1 ]   (150,350) -- (230,350) ;
    \draw [color={rgb, 255:red, 208; green, 2; blue, 27 }  ,draw opacity=1 ][line width=3]  [dash pattern={on 3.38pt off 3.27pt}]  (150,370) -- (240,370) ;
    
    \draw (6,2) node [anchor=north west][inner sep=0.75pt]  [font=\Large] [align=left] {$\displaystyle R$};
    \draw (251,15) node [anchor=north west][inner sep=0.75pt]  [align=left] {$\displaystyle \ldots $};
    \draw (679,15) node [anchor=north west][inner sep=0.75pt]   [align=left] {$\displaystyle \ldots $};
    \draw (231,2) node [anchor=north west][inner sep=0.75pt]  [font=\Large] [align=left] {$\displaystyle \mathcal{C}$};
    \draw (1,45) node [anchor=north west][inner sep=0.75pt]   [align=left] {$\displaystyle \text{BWT}_{R}$};
    \draw (78,45) node [anchor=north west][inner sep=0.75pt]   [align=left] {$\displaystyle \text{SA}_{R}$};
    \draw (123,72.87) node [anchor=north west][inner sep=0.75pt]  [rotate=-89.62] [align=left] {$\displaystyle \ldots $};
    \draw (123,152.87) node [anchor=north west][inner sep=0.75pt]  [rotate=-89.62] [align=left] {$\displaystyle \ldots $};
    \draw (123,271.01) node [anchor=north west][inner sep=0.75pt]  [rotate=-89.62] [align=left] {$\displaystyle \ldots $};
    \draw (52,271.01) node [anchor=north west][inner sep=0.75pt]  [rotate=-89.62] [align=left] {$\displaystyle \ldots $};
    \draw (52,72.87) node [anchor=north west][inner sep=0.75pt]  [rotate=-89.62] [align=left] {$\displaystyle \ldots $};
    \draw (52,152.87) node [anchor=north west][inner sep=0.75pt]  [rotate=-89.62] [align=left] {$\displaystyle \ldots $};
    \draw (288,12) node [anchor=north west][inner sep=0.75pt]   [align=left] {\textbf{\$}};
    \draw (651,12) node [anchor=north west][inner sep=0.75pt]   [align=left] {\textbf{\$}};
    \draw (117,222) node [anchor=north west][inner sep=0.75pt]   [align=left] {$\displaystyle i$};
    \draw (111,245) node [anchor=north west][inner sep=0.75pt]  [font=\scriptsize] [align=left] {$\displaystyle i$+1};
    \draw (221,221) node [anchor=north west][inner sep=0.75pt]  [color={rgb, 255:red, 52; green, 149; blue, 129 }  ,opacity=1 ] [align=left] {\textbf{C}};
    \draw (221,241) node [anchor=north west][inner sep=0.75pt]  [color={rgb, 255:red, 52; green, 149; blue, 129 }  ,opacity=1 ] [align=left] {\textbf{N}};
    \draw (411,22) node [anchor=north west][inner sep=0.75pt]  [color={rgb, 255:red, 52; green, 149; blue, 129 }  ,opacity=1 ] [align=left] {\textbf{G}};
    \draw (501,72) node [anchor=north west][inner sep=0.75pt]  [color={rgb, 255:red, 52; green, 149; blue, 129 }  ,opacity=1 ] [align=left] {\textbf{A}};
    \draw (457,161) node [anchor=north west][inner sep=0.75pt]  [color={rgb, 255:red, 208; green, 2; blue, 27 }  ,opacity=1 ] [align=left] {\textbf{A}};
    \draw (397,102) node [anchor=north west][inner sep=0.75pt]  [color={rgb, 255:red, 208; green, 2; blue, 27 }  ,opacity=1 ] [align=left] {\textbf{A}};
    \draw (317,22) node [anchor=north west][inner sep=0.75pt]  [color={rgb, 255:red, 144; green, 19; blue, 254 }  ,opacity=1 ] [align=left] {\textbf{T}};
    \draw (44,222) node [anchor=north west][inner sep=0.75pt]  [color={rgb, 255:red, 208; green, 2; blue, 27 }  ,opacity=1 ] [align=left] {\textbf{A}};
    \draw (591,141) node [anchor=north west][inner sep=0.75pt]  [color={rgb, 255:red, 52; green, 149; blue, 129 }  ,opacity=1 ] [align=left] {\textbf{T}};
    \draw (124,321) node [anchor=north west][inner sep=0.75pt]  [color={rgb, 255:red, 208; green, 2; blue, 27 }  ,opacity=1 ] [align=left] {\textbf{A}};
    \draw (124,341) node [anchor=north west][inner sep=0.75pt]  [color={rgb, 255:red, 144; green, 19; blue, 254 }  ,opacity=1 ] [align=left] {\textbf{T}};
    \draw (124,361) node [anchor=north west][inner sep=0.75pt]  [color={rgb, 255:red, 208; green, 2; blue, 27 }  ,opacity=1 ] [align=left] {\textbf{A}};
    \draw (131,382) node [anchor=north west][inner sep=0.75pt]  [font=\footnotesize] [rotate=-89.62][align=left] {$\displaystyle \ldots $};
    \draw (97,22) node [anchor=north west][inner sep=0.75pt]  [color={rgb, 255:red, 208; green, 2; blue, 27 }  ,opacity=1 ] [align=left] {\textbf{A}};
    \draw (27,2) node [anchor=north west][inner sep=0.75pt]  [color={rgb, 255:red, 208; green, 2; blue, 27 }  ,opacity=1 ] [align=left] {\textbf{A}};
    \draw (122,2) node [anchor=north west][inner sep=0.75pt]  [color={rgb, 255:red, 52; green, 149; blue, 129 }  ,opacity=1 ] [align=left] {\textbf{N}};
    \draw (191,22) node [anchor=north west][inner sep=0.75pt]  [color={rgb, 255:red, 52; green, 149; blue, 129 }  ,opacity=1 ] [align=left] {\textbf{C}};
    
    \end{tikzpicture}
    }
    \caption{Example showcasing the proposed heuristic. Solid coloured lines under $\Ch$ are matching factors for $\ihead$s, while solid black lines are matching factors for suffixes inbetween $\ihead$s. Dotted red lines are for suffixes inbetween $\ihead$s that share the same insert-point as the first solid red line (an $\ihead$). Suffixes inbetween $\ihead$s share the same preceding character (red A) with the reference $R$, while the red-coloured $\ihead$ is preceded by a different character (purple T). Zooming in on the insert-bucket, we see that it can happen that the purple T goes between the red As, breaking what would have been a run of only red As.}
    \label{fig:example-head-vs-suf}
\end{figure}

\begin{lemma}\label{lemma:same-preceding-character}
Let $\Suf_i$ and $\Suf_j$ be two suffixes of $\Ch$. If $\ip(i) = \ip(j)$ and the two suffixes are not the start of an $\ihead$, then $\Suf_i$ and $\Suf_j$ are preceded by the same character $c = R[\SA[\ip(i)]-1]$, i.e.\ $\Ch[i-1] = \Ch[j-1] = c = R[\SA[\ip(i)]-1]$.
\end{lemma}

\begin{proof}
By assumption we know that $\ip(i) = \ip(j)$, therefore $\SA[\ip(i)] = \SA[\ip(j)]$. Because $\Suf_i$ and $\Suf_j$ are not the starting positions of any insert-head, it is true that $\SA[\ip(i-1)] = \SA[\ip(i)]-1$ and $\SA[\ip(j-1)] = \SA[\ip(j)]-1$. Therefore, $\SA[\ip(i-1)] = \SA[\ip(j-1)]$ and also $\ip(i-1) = \ip(j-1)$. Since $U_{i-1},U_{j-1} \neq \epsilon$ it follows that the first character of $U_{i-1}$ and $U_{j-1}$ is the same.
\end{proof}

While computing the $\eCMS$ of $\Ch$ w.r.t.\ $R$, we can simultaneously count how many suffixes fall in each insert-bucket. We recall that we can have at most $|R|$ insert-buckets, so the size of the array of counters called {\em bucket-counters} is $|R|\log N$ bits.
By Lemma~\ref{lemma:same-preceding-character} we know that suffixes in the same insert-bucket have different preceding character only if one of them is an $\ihead$. By scanning again the collection, we just need to count how many suffixes belonging to the same insert-bucket come before each $\ihead$. In a sense, $\ihead$s work as run boundaries inside their insert-bucket, because they are preceded by a character that is different from the one preceding other non-$\ihead$ suffixes.
Therefore, we only need an additional counter for each $\ihead$ to keep track of this quantity. We will store the counters in an array called {\em head-counters}. Inside a given insert-bucket we already know the total order of $\ihead$s, because we have sorted the whole set $K$ after the computation of $\eCMS$, as mentioned in Section~\ref{subsec:compare-suf}. 

Since we are scanning $\Ch$ left-to-right, we know $\MS[i]$ for every suffix, without the need of using predecessor queries as we explain next.
By saving the $\eCMS$ in text order, we start from the first $\ihead$ $k_1$. Every suffix $i$ before the starting position of $k_2$ have $MS[i] = (q_1 + (i - j_1), \ell_1 - (i - j_1))$, where $j_1 = 1$ is the starting position of $k_1$. Then, when we reach the starting position of $k_2$, we just have to repeat this procedure until every couple of $\ihead$s has been processed. We will compare suffix $i$ only with $\ihead$s stored in the corresponding insert-bucket, therefore we do not need to perform any predecessor query. Ultimately, the comparisons are made in constant time.
Given $\Suf_i$, we can perform a binary search in the bucket corresponding to $\ip(i)$ taking $\Oh(\log K_i)$ time, where $K_i$ is the set of $\ihead$s in the bucket with $\ip(i)$. After finding the correct index using Lemma~\ref{lemma:sufcomp_2} and, if necessary, resorting to the rank of the sorted $\ihead$s, we increment the counter for that $\ihead$. The array of {\em head-counters} takes $\kappa\log N$ bits of space.

Building the $\BWT$ of $\Ch$ is then just a matter of interleaving bucket-counters and head-counters. For $1 \leq i \leq |R|$, let $x = \text{bucket-counter}[i]$ be the number of suffixes in that bucket. If no $\ihead$s are present in the bucket,  write $c = R[\SA[i] - 1]$ in the output $\BWT(\Ch)$ $x$ times. Otherwise, if at least one $\ihead$ is in the bucket, for each head-counter in the current insert-bucket write $c$ repeated as many times as indicated in the head-counter. Then, after the head-counter is processed, write the character that precedes the $\ihead$ itself, namely $\Ch[j-1]$, where $j$ is the starting position of the $\ihead$. Each time we write a character from either the head-counter or the head itself, we subtract one from $x$. If at the end of this procedure, $x$ is not equal to 0, it means we still need to write that number of $c$ characters in the output $\BWT$. This is because this amount of suffixes was bigger than any $\ihead$ in their insert-bucket. 

The main procedure is outlined in Algorithm~\ref{algo:CMS-BWT} and the running time and space consumption are reported in Proposition~\ref{prop:CMS-BWT}. A full example can be found in Figure~\ref{fig:example-cms-bwt}.

\begin{proposition}\label{prop:CMS-BWT}
Given $R$ and $\Ch$, we can compute $\BWT(\Ch)$ in $\Oh(N\log\kappa + N\log|R| + |R|)$ time and $\Oh(\kappa + |R|)$ space.
\end{proposition}

\begin{proof}
Computing all data structures for $R$ can be done in linear time and space in $|R|$. Computing the $\eCMS$ of $\Ch$ takes $\Oh(N\log|R|)$ time and $\Oh(|R|)$ space using the approach described in Section~\ref{subsec:compute-eCMS}. The computation of $\BWT(\Ch)$ is bounded by the time of counting how many suffixes are smaller than each $\ihead$ in a bucket with the same $\ip(i)$. More specifically, $\sum_{1\leq i \leq |R|}B_i\log K_i \leq |C|\log\kappa$, where $B_i$ is the set of suffixes belonging to insert-bucket $i$ and $K_i$ the set of $\ihead$s within the same insert-bucket $i$. The space consumption is dominated by the number of $\ihead$s and the size of the data structures on $R$.
\end{proof}

Because we are working with highly similar strings, we expect to have few $\ihead$s, having long matches between any string of $\Ch$ and $R$. This makes the sorting part of $\ihead$s very fast in practice, due to $\kappa$ being small. Also, the process of binary searching is conducted bucket by bucket, so the number of heads in the same bucket is expected to be smaller than $\kappa$.

Moreover, if the $\ihead$s are concentrated in a few insert-buckets we can entirely skip the computation for each bucket without $\ihead$s. More information on real-life datasets related to this insight can be found in Section~\ref{subsec:datasets}.

\begin{algorithm}[ht]
\caption{CMS-BWT} \label{algo:CMS-BWT}
\DontPrintSemicolon
\KwIn{reference $R$, collection $\Ch$}
\KwOut{$\BWT(\Ch)$}

compute data structures on $R$ \;
compute $\eCMS$ of $\Ch$ w.r.t.\ $R$ and count how many suffixes we have in each insert-bucket \;
sort $\eCMS$ \;
\ForEach{suffix in an insert-bucket containing $\ihead$s}{ \label{algo:for-loop}
    binary search the correct position and increase head-counter \;
} 
output the $\BWT$ of $\Ch$ by interleaving head-counters and bucket-counters \;

\end{algorithm}

\begin{figure}
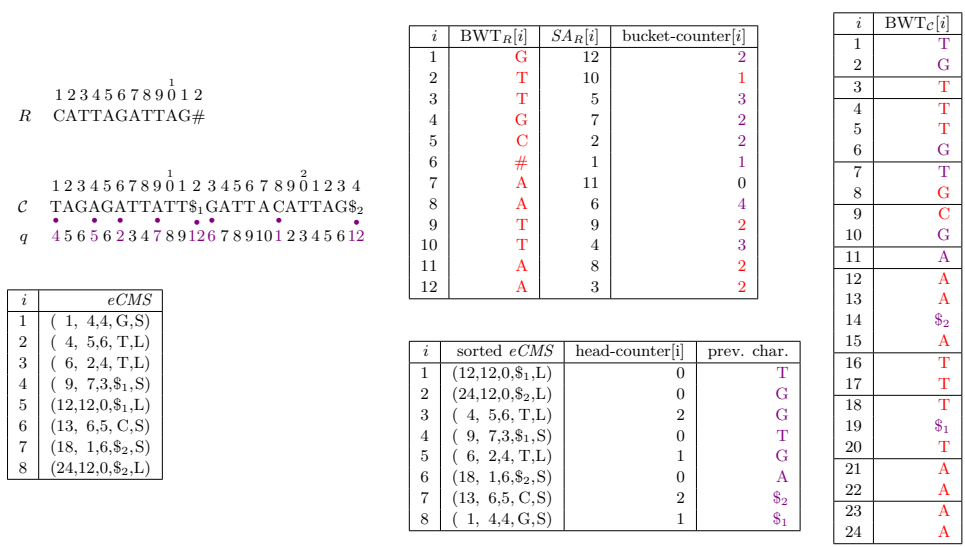

    \centering
    \resizebox{.94\textwidth}{!}{
    \begin{tabular}{lc c}
        \begin{tabular}{l}
            \begin{tabular}{rc@{}c@{}c@{}c@{}c@{}c@{}c@{}c@{}c@{}c@{}c@{}c@{}}
                    &  1 & 2 & 3 & 4 & 5 & 6 & 7 & 8 & 9 & $\overset{1}{0}$ & 1 & 2 \\
                 $R$ & C & A & T & T & A & G & A & T & T & A & G & \# \\
            \end{tabular}  \\
            \\\\
            \begin{tabular}{rc@{}c@{}c@{}c@{}c@{}c@{}c@{}c@{}c@{}c@{}c@{}c@{}c@{}c@{}c@{}c@{}c@{}c@{}c@{}c@{}c@{}c@{}c@{}c@{}}
                 & 1 & 2 & 3 & 4 & 5 & 6 & 7 & 8 & 9 & $\overset{1}{0}$ & 1 & 2 & 3 & 4 & 5 & 6 & 7 & 8 & 9 & $\overset{2}{0}$ & 1 & 2 & 3 & 4 \\
                 $\Ch$ & $\underset{\textcolor{violet}{\bullet}}{\text{T}}$ & A & G & $\underset{\textcolor{violet}{\bullet}}{\text{A}}$ & G & $\underset{\textcolor{violet}{\bullet}}{\text{A}}$ & T & T & $\underset{\textcolor{violet}{\bullet}}{\text{A}}$ & T & T & $\underset{\textcolor{violet}{\bullet}}{\$_1}$ & $\underset{\textcolor{violet}{\bullet}}{\text{G}}$ & A & T & T & A & $\underset{\textcolor{violet}{\bullet}}{\text{C}}$ & A & T & T & A & G & $\underset{\textcolor{violet}{\bullet}}{\$_2}$ \\
                 $q$ & \textcolor{violet}{4} & 5 & 6 & \textcolor{violet}{5} & 6 & \textcolor{violet}{2} & 3 & 4 & \textcolor{violet}{7} & 8 & 9 & \textcolor{violet}{12} & \textcolor{violet}{6} & 7 & 8 & 9 & 10 & \textcolor{violet}{1} & 2 & 3 & 4 & 5 & 6 & \textcolor{violet}{12} \\
            \end{tabular}  \\ 
            \\\\
            \begin{tabular}{|r|r@{}r@{}r@{}r@{}r@{}r|}
                \hline
                $i$ & \multicolumn{6}{r|}{$\eCMS$}\\
                \hline
                1 & ( & 1, & 4, & 4, & G, & S) \\
                2 & ( & 4, & 5, & 6, & T, & L) \\
                3 & ( & 6, & 2, & 4, & T, & L) \\
                4 & ( & 9, & 7, & 3, & $\$_1$, & S) \\
                5 & ( & 12, & 12, & 0, & $\$_1$, & L) \\
                6 & ( & 13, & 6, & 5, & C, & S) \\
                7 & ( & 18, & 1, & 6, & $\$_2$, & S) \\
                8 & ( & 24, & 12, & 0, & $\$_2$, & L) \\
                \hline
            \end{tabular}
        \end{tabular}

        &  
    
        \begin{tabular}{l}
            \begin{tabular}{|r|r|r|r|}
                 \hline
                 $i$ & $\BWT_R[i]$ & $\SA_R[i]$ & bucket-counter[$i$]  \\
                 \hline
                 1 & \textcolor{red}{G} & 12 & \textcolor{violet}{2} \\
                 2 & \textcolor{red}{T} & 10 & \textcolor{red}{1}\\
                 3 & \textcolor{red}{T} & 5 & \textcolor{violet}{3}\\
                 4 & \textcolor{red}{G} & 7 & \textcolor{violet}{2}\\
                 5 & \textcolor{red}{C} & 2 & \textcolor{violet}{2}\\
                 6 & \textcolor{red}{\#} & 1 & \textcolor{violet}{1}\\
                 7 & \textcolor{red}{A} & 11 & 0\\
                 8 & \textcolor{red}{A} & 6 & \textcolor{violet}{4}\\
                 9 & \textcolor{red}{T} & 9 & \textcolor{red}{2}\\
                 10 & \textcolor{red}{T} & 4 & \textcolor{violet}{3}\\
                 11 & \textcolor{red}{A} & 8 & \textcolor{red}{2}\\
                 12 & \textcolor{red}{A} & 3 & \textcolor{red}{2}\\
                 \hline
            \end{tabular} \\
            \\\\
            \begin{tabular}{|r|r@{}r@{}r@{}r@{}r@{}r|r|r|}
                \hline
                $i$ & \multicolumn{6}{r|}{sorted $\eCMS$} & head-counter[i] & prev. char. \\
                \hline
                1 & ( & 12, & 12, & 0, & $\$_1$, & L) & 0 & \textcolor{violet}{T}\\
                2 & ( & 24, & 12, & 0, & $\$_2$, & L) & 0 & \textcolor{violet}{G}\\
                3 & ( & 4, & 5, & 6, & T, & L) & 2 & \textcolor{violet}{G}\\
                4 & ( & 9, & 7, & 3, & $\$_1$, & S) & 0 & \textcolor{violet}{T}\\
                5 & ( & 6, & 2, & 4, & T, & L) & 1 & \textcolor{violet}{G}\\
                6 & ( & 18, & 1, & 6, & $\$_2$, & S) & 0 & \textcolor{violet}{A}\\
                7 & ( & 13, & 6, & 5, & C, & S) & 2 & \textcolor{violet}{$\$_2$}\\
                8 & ( & 1, & 4, & 4, & G, & S) & 1 & \textcolor{violet}{$\$_1$}\\
                \hline
            \end{tabular}
        \end{tabular}
    
        & 
         \begin{tabular}{|r|r|}
             \hline
             $i$ & $\BWT_\Ch[i]$ \\
             \hline
             1 & \textcolor{violet}{T} \\
             2 & \textcolor{violet}{G}\\
             \hline
             3 & \textcolor{red}{T}\\
             \hline
             4 & \textcolor{red}{T}\\
             5 & \textcolor{red}{T}\\
             6 & \textcolor{violet}{G}\\
             \hline
             7 & \textcolor{violet}{T}\\
             8 & \textcolor{red}{G}\\
             \hline
             9 & \textcolor{red}{C}\\
             10 & \textcolor{violet}{G}\\
             \hline
             11 & \textcolor{violet}{A}\\
             \hline
             12 & \textcolor{red}{A}\\
             13 & \textcolor{red}{A}\\
             14 & \textcolor{violet}{$\$_2$}\\
             15 & \textcolor{red}{A}\\
             \hline
             16 & \textcolor{red}{T}\\
             17 & \textcolor{red}{T}\\
             \hline
             18 & \textcolor{red}{T}\\
             19 & \textcolor{violet}{$\$_1$}\\
             20 & \textcolor{red}{T}\\
             \hline
             21 & \textcolor{red}{A}\\
             22 & \textcolor{red}{A}\\
             \hline
             23 & \textcolor{red}{A}\\
             24 & \textcolor{red}{A}\\
             \hline
        \end{tabular}   
        \\
    \end{tabular}
    }
    
    \caption{Example of the construction of $\BWT(\Ch)$. The colour purple is used to indicate a relationship with $\ihead$s, whereas red is used to indicate a relationship with $R$ and $\BWT_R$. 
    On the left, under $\Ch$ we mark $\ihead$s with a purple circle and highlight with the same colour $q_j$ when $j$ is an $\ihead$. 
    In the middle part of the figure, entries of bucket-counter highlighted in red contains a positive number, meaning that no $\ihead$ is contained within that bucket. On the other hand, entries coloured in purple tell us that we have at least one $\ihead$ in that bucket.
    On the right, we show the full $\BWT$ of $\Ch$, where we use the same colour code. We also show with horizontal lines insert-buckets, highlighting how we interleaved information from bucket-counters and head-counters.
    }
    \label{fig:example-cms-bwt}
\end{figure}

\section{Implementation details}\label{sec:implementation}
The algorithm starts by first augmenting the reference with characters that occur in $\Ch$ but not in $R$ so that we have a well-defined insert point for each suffix of the collection.

Then, we use \libsais~\cite{libsais} to build $\SA_{R}$ and $\LCP_{R}$. We chose this tool because it has been experimentally proven to be one of the fastest tools for general-purpose suffix array construction. For the $\LCP$ array it uses the $\Phi$ method~\cite{KMP09}. The data structure for $\PSV$-$\NSV$ queries on the $\LCP_{R}$ is based on the work of C{\'a}novas and Navarro~\cite{CN10}. 

For sorting the $\eCMS$, we first rename each $\ihead$ with a metacharacter based on the rank of the partial lexicographic order of the substrings associated with each $\ihead$. Then, by rearranging these metacharacters in text order we use again \libsais\ to compute the suffix array of this metacharacter string. 

When profiling the implementation, we found that the number of distinct $\ihead$s, i.e.\ the number of different tuples in $K$, grows even slower than the total set. For example, looking at the dataset consisting of 330 copies of Human Chromosome 19 described in Section~\ref{subsec:datasets}, we have only 4,355,600 unique $\ihead$s versus $\kappa = 174,532,868$. Moreover, when performing binary search comparisons, more than 60\% of the total number of comparisons were resolved by comparing the length plus $x_i$ information stored in the $\eCMS$. Combining these two insights led us to another heuristic based on a two-layered binary search. First, we compare the length information $\ell_i = \ell_{k_i} - (i - j)$ along with $x_{k_i}$ of a suffix $i$ with $\ihead$ $\pred_{K}(i) = k_i$ starting at position $j$ with the set of unique $\ihead$s of its insert-bucket.
Then, if the pair $\ell_i$ and $x_{k_i}$ is different from any other $\ihead$ we increment the counter for the $\ihead$ pointed to by this first binary search. Otherwise, we have to refine the search by comparing $\Suf_i$ with the whole set of $\ihead$s having $\ip(i)$, $\ell = \ell_i$ and $x_{k_i}$. This technique led to a speedup in the binary search phase of between 10\% and 20\%.

To avoid continuous cache misses due to loading different subsets of $\ihead$s with different insert points during binary searching, we put a number of suffixes in a buffer divided into insert-buckets. After the buffer is at its full capacity, we proceed to process in bulk suffixes in the same insert-bucket, easing the loading in cache of subsets of $\ihead$s. For all of our experiments, we set this buffer to 2GB, but it can be arbitrarily chosen by the user.

Lastly, we also implemented a variant of \CMSBWT\ trading off space for running time. This was achieved by writing to disk some of the data structures involved in different phases of the algorithm. This version saves roughly a third of the space used by the non-memory-saving implementation.

\section{Experiments}\label{sec:experiments}
We implemented our algorithm for computing the $\BWT$ in C++. Our implementation, \CMSBWT, is available at~\url{https://github.com/fmasillo/CMS-BWT}. The experiments were conducted on a desktop equipped with 64GB of RAM DDR4-3200MHz and an Intel(R) Core(R) i9-11900 @ 2.50GHz (with turbo speed @ 5GHz) with 16 MB of cache. The operating system was Ubuntu 22.04 LTS, the compiler used was {\tt g++} version 11.3.0 with options {\tt -std=c++20 -O3 -funroll-loops -march=native} enabled.

\subsection{Tools compared}\label{subsec:tools}
We compared two different implementations of \CMSBWT\ (simple and memory-saving) to the following four tools:
\begin{enumerate}
    \item \bigBWT~\cite{BoucherGKLMM19}, a tool computing both the $\BWT$ and the suffix array. It is specifically made for highly repetitive data. We used the default parameters ({\tt -w = 10}, {\tt -p = 100}) and the {\tt -f} flag to parse fasta files as input. This tool outputs the whole $\BWT$.
    \item \rPFP~\cite{Oliva2023}, is a tool improving on plain PFP. It has been shown to be both faster and to use less space than \bigBWT\ for big enough dataset sizes. We run the experiments using {\tt -}{\tt -bwt-only -}{\tt -w1 10 -}{\tt-w2 5} as flags. This tool outputs the run-length encoded $\BWT$.
    \item \grlBWT~\cite{Diaz-DominguezN22}, a tool computing the BCR $\BWT$~\cite{BauerCR13} again with a focus on highly repetitive data. We used the default parameters. This tool outputs the run-length encoded $\BWT$. 
    \item \ropeBWT~\cite{Li14a}, a highly optimized tool to compute the BCR $\BWT$ on DNA data. We used the flag {\tt -R} to skip the reverse complement. We also compare the effect of adding the {\tt -P} flag, which limits the software to execute in single-threaded mode at all times. This tool outputs the whole $\BWT$.
\end{enumerate}

\subsection{Datasets}\label{subsec:datasets}
In our experiments, we used two publicly available datasets. The first dataset, called {\tt chr19} contains copies of the Human Chromosome 19 from the 1000 Genomes Project~\cite{1000genomes}. The second dataset, named {\tt sars-cov2}, consists of copies of SARS-CoV2 genomes taken from COVID-19 Data Portal~\footnote{We used the following command to download in bulk the data using the CDP File Downloader: {\tt java -jar cdp-file-downloader.jar -}{\tt  -domain=VIRAL\_SEQUENCES -}{\tt  -datatype=SEQUENCES -}{\tt  -format=FASTA -}{\tt  -location=/home/data/ -}{\tt  -email=xxx@xxx.xx -}{\tt  -protocol=FTP}}. Some additional metadata can be found in Table~\ref{tab:datasets}.

The total size of both datasets is 60GB. We took increasing prefixes of size 1GB, 10GB, 20GB, 40GB, and 60GB.

For example, looking at 20GB of {\tt chr19} data, we have around 330 copies of Human Chromosome 19. On this prefix of {\tt chr19}, we have $\ihead$s only in 6\% of the buckets. This means that around 94\% of the suffixes will not be compared against any $\ihead$, speeding up the whole process.

\begin{table}[th]
    \centering
    \begin{tabular}{|l|l|r|r|r|r|}
         \hline
         Name & Description & $\sigma$ & $r$ & no.\ of i-heads & no.\ unique i-heads \\
         & & & (20 GB) & (20 GB) & (20 GB) \\
         \hline
         {\tt chr19} & Human Chromosome 19 & 5 & 36\,723\,404 & 174\,532\,868 & 4\,355\,600 \\
         {\tt sars-cov2} & SARS-CoV2 genome & 14 & 19\,075\,277 & 253\,188\,521 & 1\,466\,183 \\
         \hline
    \end{tabular}
    \caption{Datasets used in experiments. In column 3, we specify the alphabet size $\sigma$, in column 4 the number $r$ of runs of the BWT, in column 5 the number of insert-heads, and in column 6 the number of unique insert-heads. In our experiments, we use prefixes of each dataset up to 60GB. The last three columns refer to the 20GB prefix. }
    \label{tab:datasets}
\end{table}

\subsection{Results}
As already pointed out in Section~\ref{subsec:tools}, the output of the tools can either be the full $\BWT$ or the run-length encoded $\BWT$. This can be a non-negligible time overhead. Therefore, when comparing to tools that output the whole $\BWT$ we will also write this version of the $\BWT$ to disk. On the other hand, when comparing \CMSBWT\ to \rPFP\ and \grlBWT\ we will write to disk the run-length encoded $\BWT$.

In Figures~\ref{fig:chr19-time-full}, \ref{fig:sars-cov2-time-full}, \ref{fig:chr19-time-rl} and~\ref{fig:sars-cov2-time-rl}, we report the comparison of the running time of the five tools divided by dataset and output type.  

On the {\tt chr19} dataset, we are always the fastest tool compared to other tools that output the uncompressed $\BWT$. More specifically, comparing the non-memory-saving implementation at 60GB of data, we are 57\% faster than \bigBWT\ and 10 times faster than \ropeBWT\ with and without {\tt -P}. Compared to the tools that output the run-length encoded $\BWT$, our fastest implementation is always the winner, while at 60GB, our memory-saving implementation takes 12\% more time than \rPFP. Both implementations outperform \grlBWT, e.g.\ at 60GB of data they take a fourth of the time. 

On the {\tt sars-cov2} dataset we are always the fastest tool in both settings. For example, at 60GB of data, we are faster than: \bigBWT\ by 17\%, \ropeBWT\ with {\tt -P} by 114\%, \ropeBWT\ with no {\tt -P} by 35\%, \rPFP\ by 445\% and \grlBWT\ by 53\%.

\begin{figure}
    \centering
    \begin{subfigure}{.49\textwidth}
        \centering
        \includegraphics[width=\textwidth]{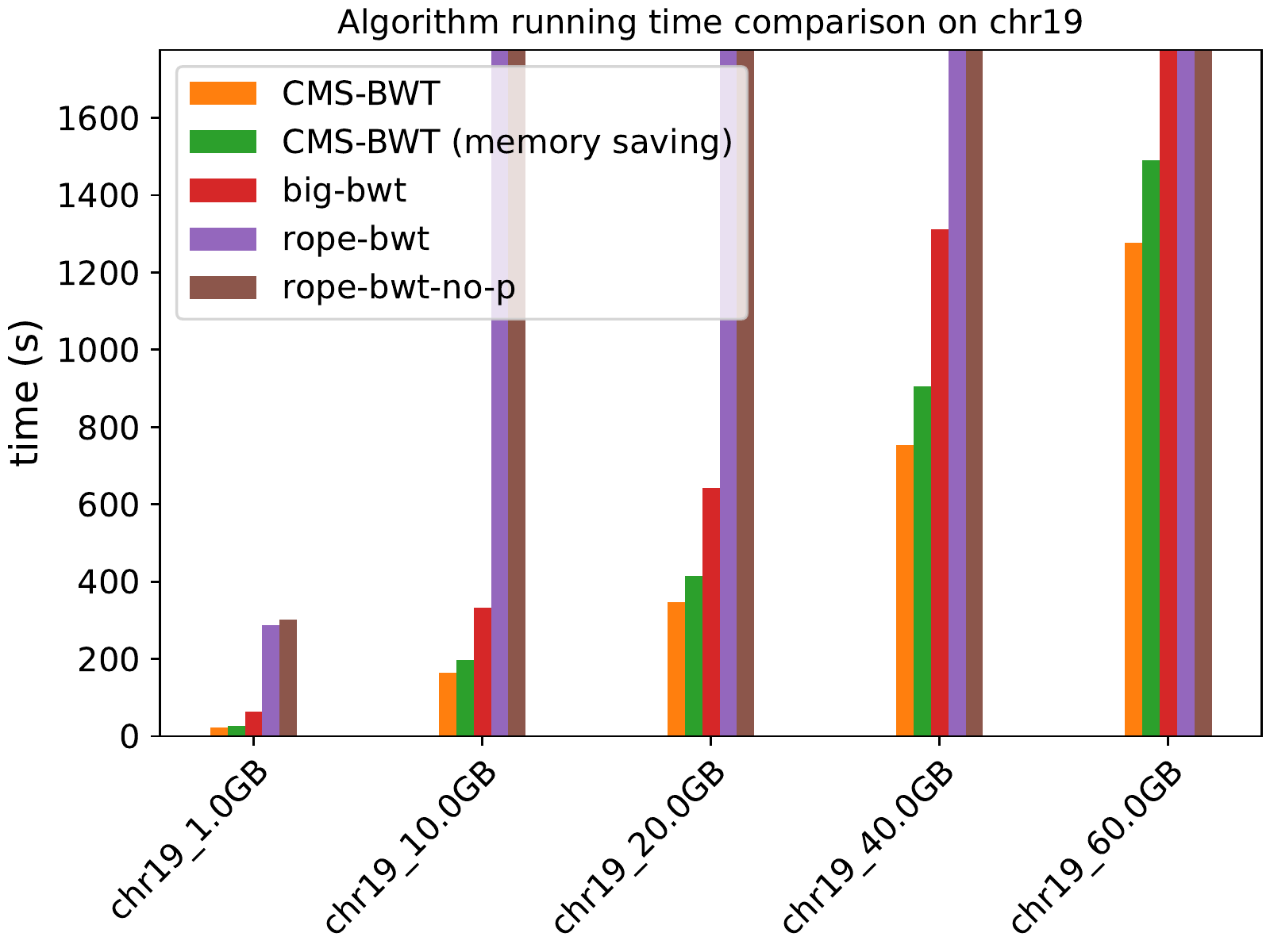}
        \caption{Running time comparison on different subsets of copies of Chromosome 19. These tools output the full $\BWT$.}
        \label{fig:chr19-time-full}
    \end{subfigure}
    \begin{subfigure}{.49\textwidth}
        \centering
        \includegraphics[width=\textwidth]{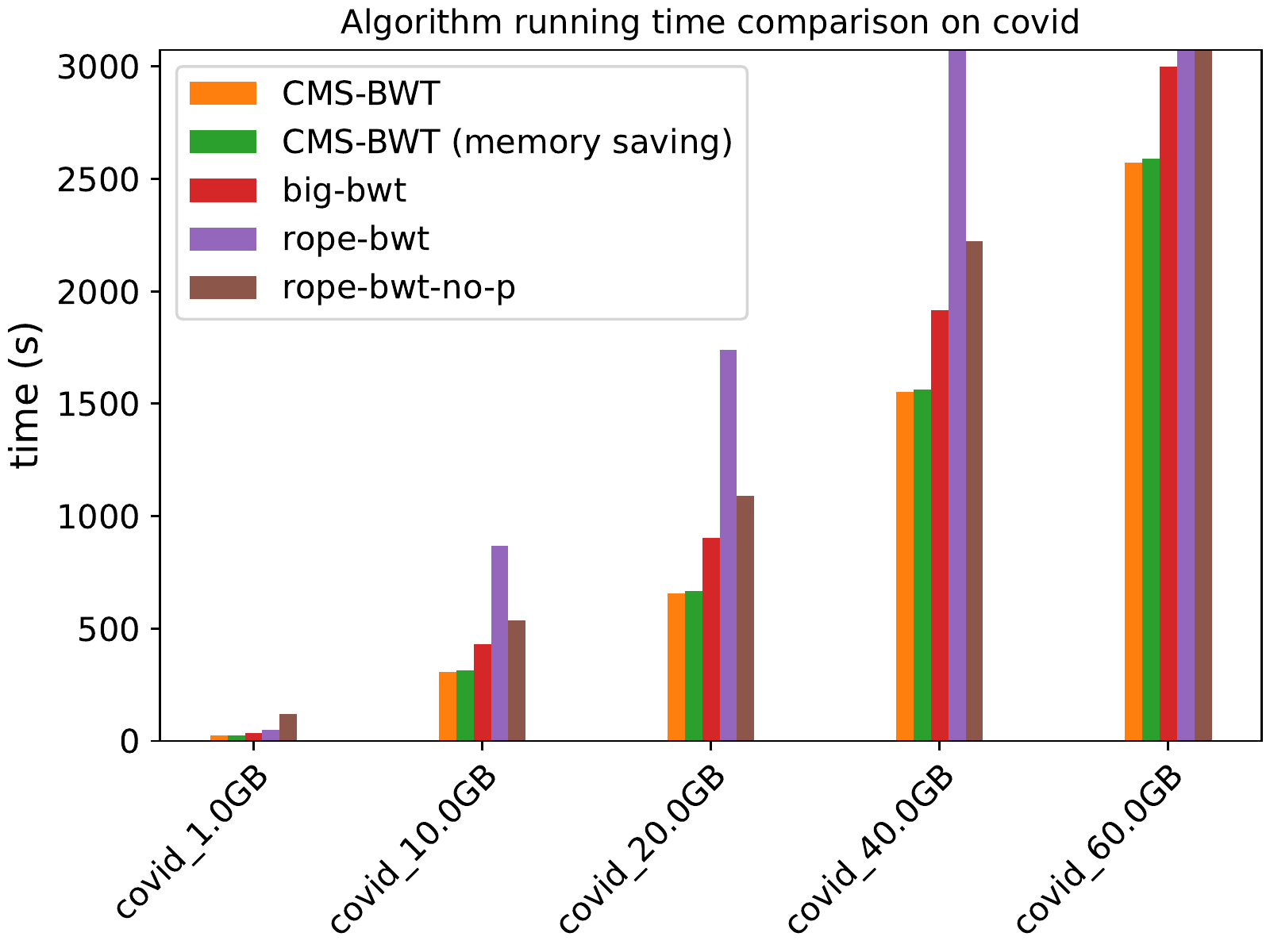}
        \caption{Running time comparison on different subsets of copies of SARS-CoV2 genomes. These tools output the full $\BWT$.}
        \label{fig:sars-cov2-time-full}
    \end{subfigure}

\end{figure}

\begin{figure}[ht]
    \centering
    \begin{subfigure}{.49\textwidth}
        \centering
        \includegraphics[width=\textwidth]{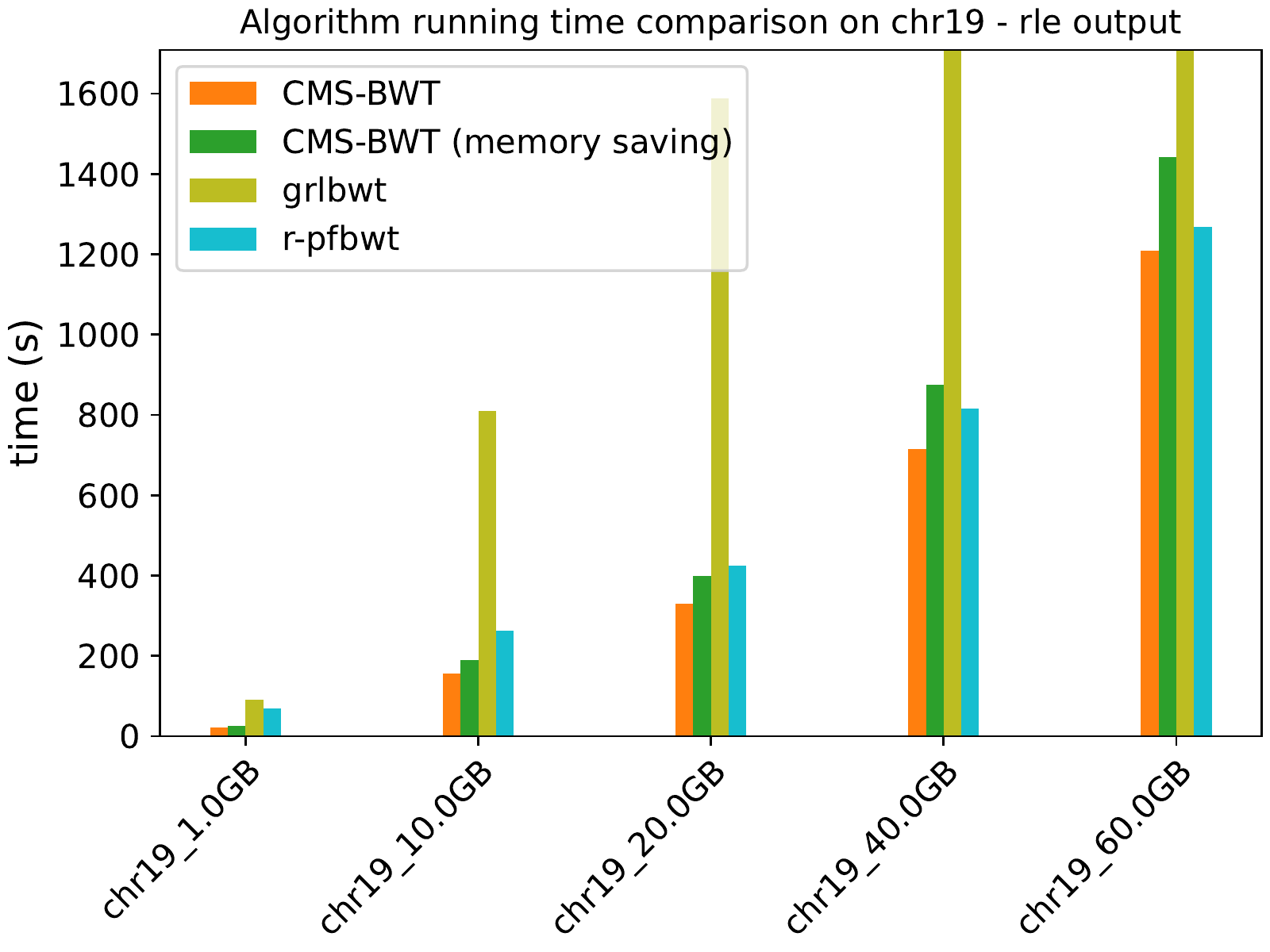}
        \caption{Running time comparison on different subsets of copies of Chromosome 19.}
        \label{fig:chr19-time-rl}
    \end{subfigure}
    \begin{subfigure}{.49\textwidth}
        \centering
        \includegraphics[width=\textwidth]{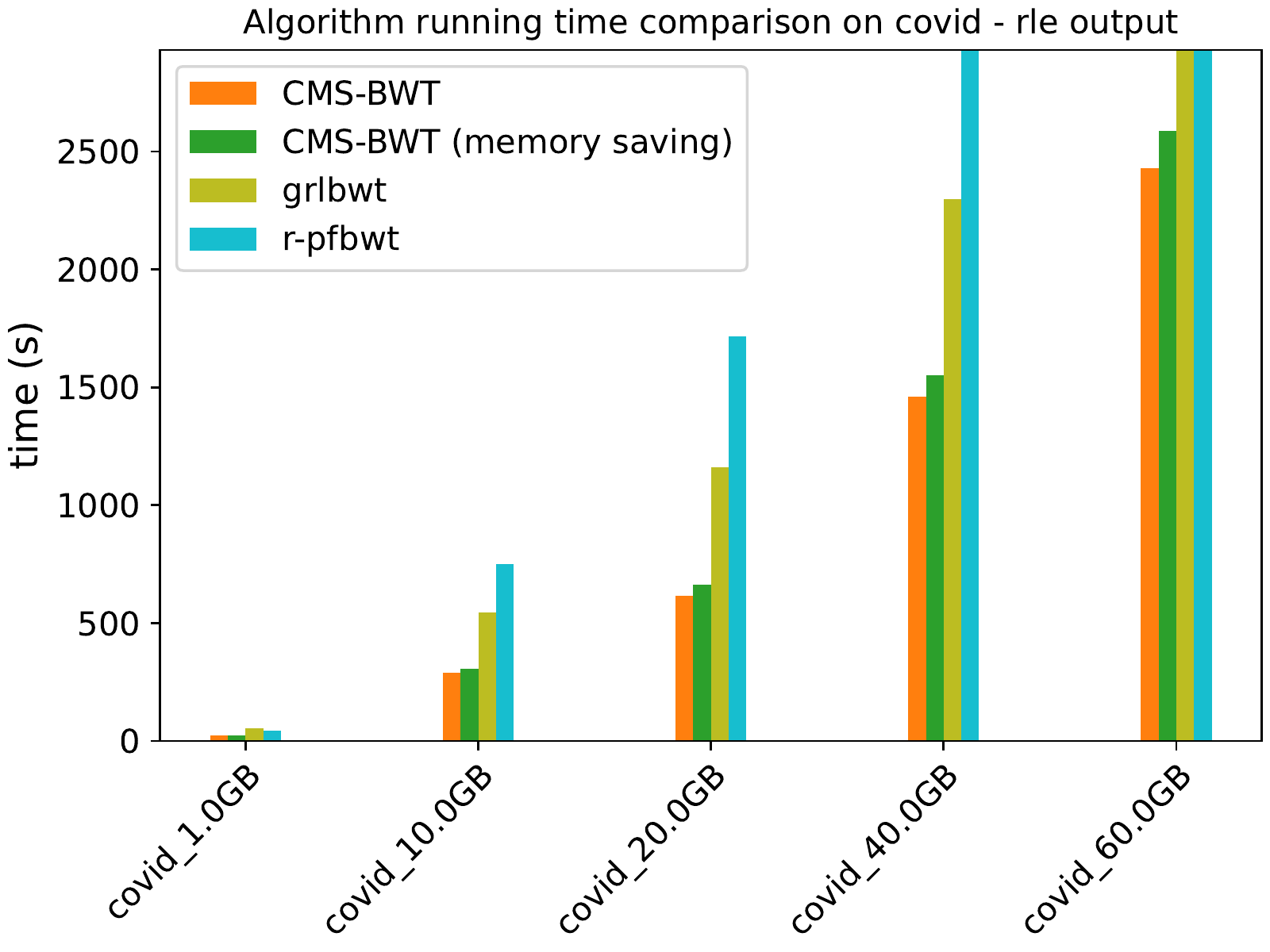}
        \caption{Running time comparison on different subsets of copies of SARS-CoV2 genomes.}
        \label{fig:sars-cov2-time-rl}
    \end{subfigure}
    \caption{Running time comparison of tools outputting the run-length encoded $\BWT$.}
\end{figure}

In Figure~\ref{fig:chr19-memory} and~\ref{fig:sars-cov2-memory} we show the memory footprint of the five tools. As one can notice, our tool has the highest memory requirement among all the other tools. However, it can be noted that the memory-saving variant of \CMSBWT\ on bigger sizes of both datasets requires always less than half of the input size in space.

\begin{figure}[ht]
    \centering
    \begin{subfigure}{.49\textwidth}
        \includegraphics[width=\textwidth]{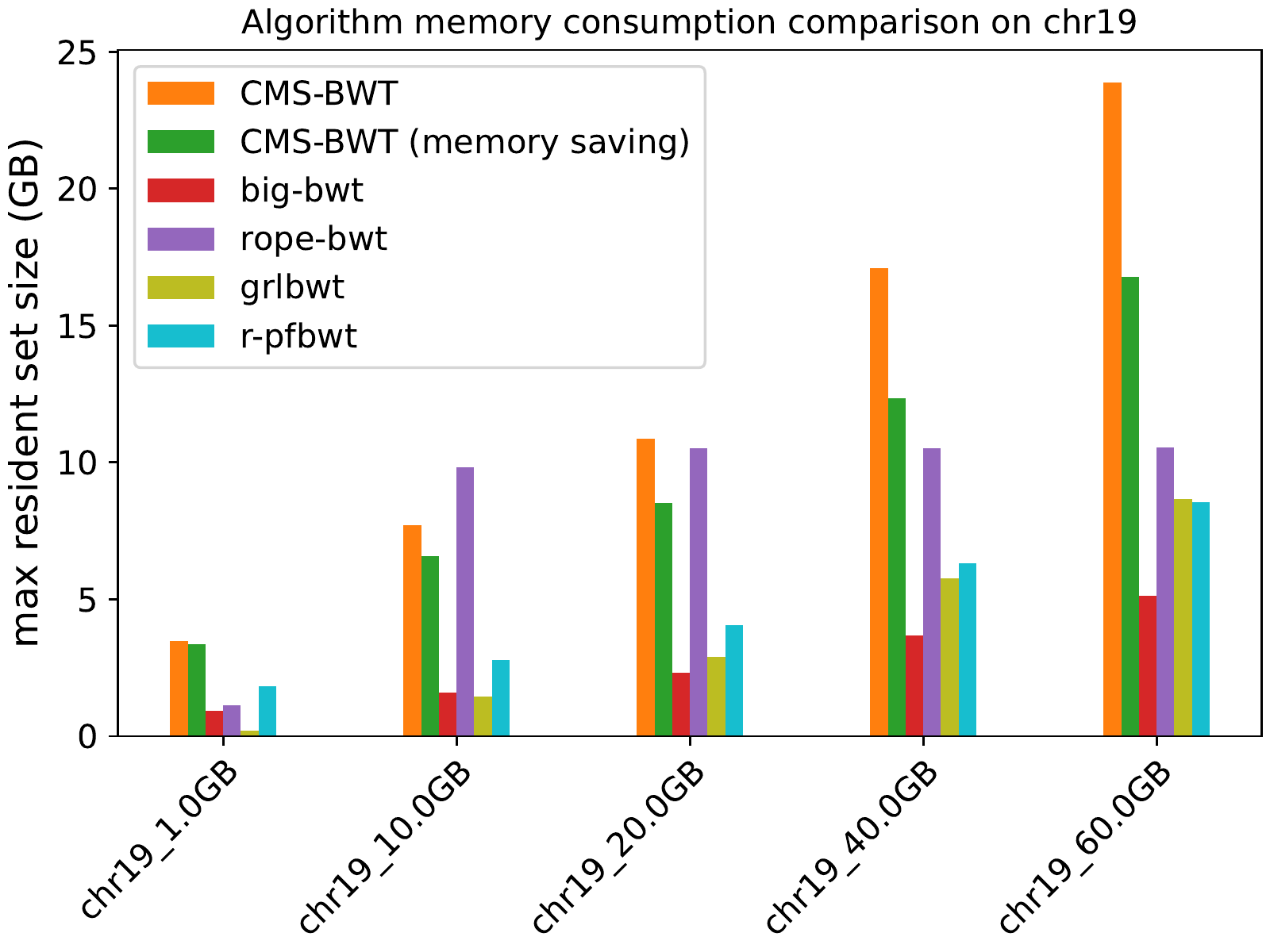}
        \caption{Comparison of tools on different subsets of copies of Chromosome 19.}
        \label{fig:chr19-memory}
    \end{subfigure}
    \begin{subfigure}{.49\textwidth}
        \includegraphics[width=\textwidth]{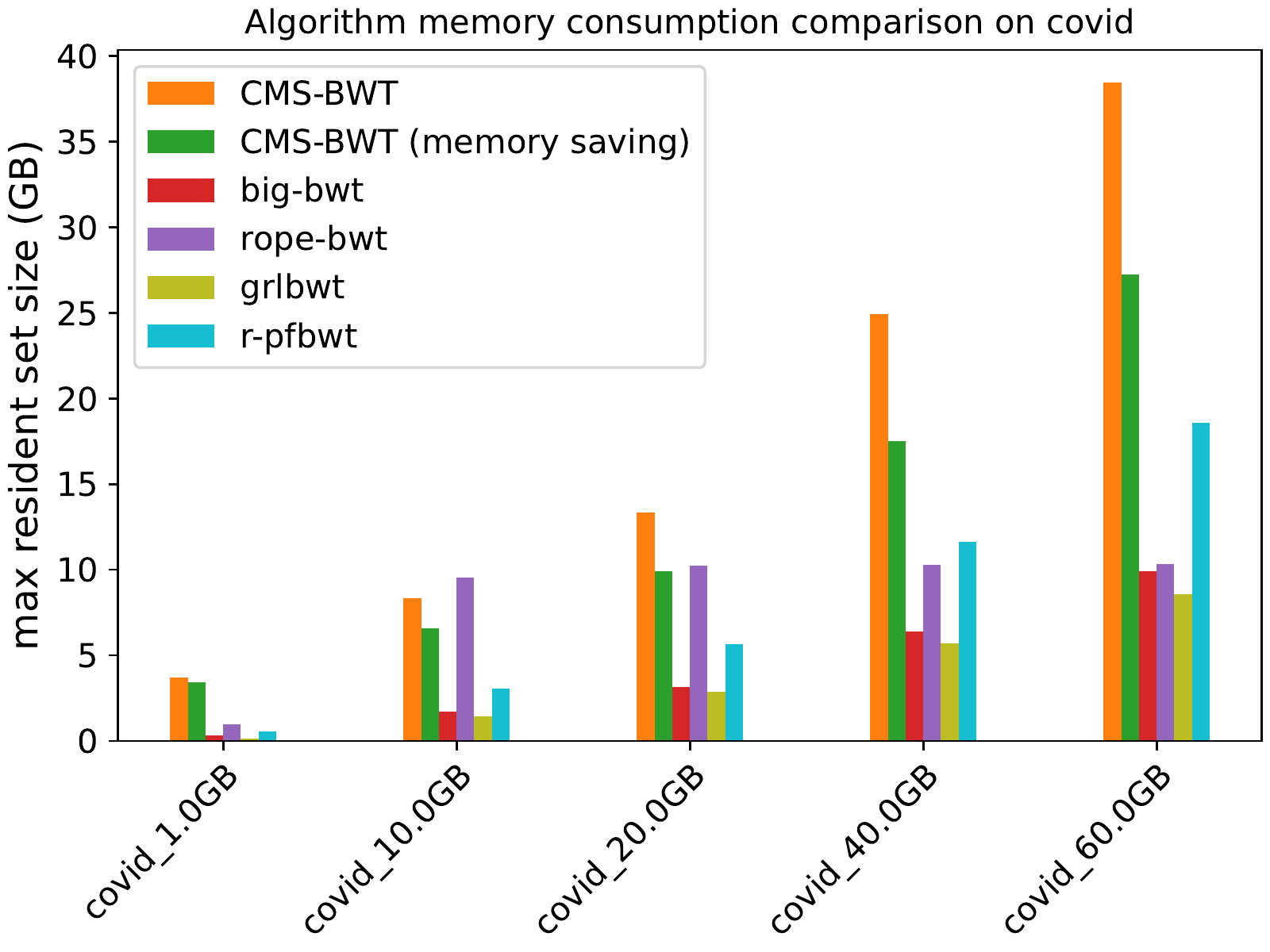}
        \caption{Comparison of tools on different subsets of SARS-CoV2 genomes.}
        \label{fig:sars-cov2-memory}
    \end{subfigure}
    \caption{Peak memory measured as maximum resident set size in GB.}
\end{figure}

\section{Conclusions}\label{sec:conclusions}
We presented a new algorithm for constructing the $\BWT$ of a collection of highly similar strings in compressed space. An experimental evaluation of two different implementations shows that our algorithm is competitive with state-of-the-art tools. Most of the time, both implementations outperform the other tools in terms of running time, but they are also the heaviest w.r.t.\ space consumption.

Future work will focus on parallelizing the implementation to allow taking advantage of multicore CPUs that are widespread nowadays. It is fairly straightforward to assign distinct sequences to a pool of multiple threads to compute the matching statistics. Another phase that would directly benefit from multi-threading is the for-loop at line~\ref{algo:for-loop} in Algorithm~\ref{algo:CMS-BWT}. With careful handling of locks for each head-counter, this is easily parallelizable, dividing the for-loop into equal parts.

We are also working on extending our algorithm to incorporate the computation of $\SA$-samples. This will allow us to build the $r$-index. With careful implementation, our tool can be extended to compute $\SA$-samples along with bucket- and head-counters, without changing both time and space bounds given in Proposition~\ref{prop:CMS-BWT}.

\bibliography{refs}

\end{document}